%% file: DISC2016new.tex
\documentclass[letterpaper, 11pt]{llncs}

\usepackage[margin=1in]{geometry}
\usepackage[bookmarks, colorlinks=true, plainpages = false, citecolor = blue,linkcolor=red,urlcolor = blue, filecolor = blue]{hyperref}
\usepackage{url}

\usepackage{amsmath,amsfonts,amssymb,bm, verbatim,mathtools}
\usepackage[linesnumbered,ruled]{algorithm2e}

\usepackage{graphicx,appendix}
\usepackage[usenames, dvipsnames]{color}

\usepackage{etoolbox}
\usepackage{tikz}
\usepackage{xr,xspace}
\usepackage{todonotes}
\usepackage{paralist}
\usepackage{caption}

\newtheorem{assumption}{Assumption}

\newcommand{\mycomment}[1]{}
\newcommand{\eproof}{\hfill $\Box$}

\usepackage{xspace,prettyref}
\newcommand{\diverge}{\to\infty}

\newcommand{\ones}{\mathbf 1}
\newcommand{\zeros}{\mathbf 0}
\newcommand{\reals}{{\mathbb{R}}}

\newcommand{\eexp}{{\rm e}}

\newcommand{\expect}[1]{\mathbb{E}\left[ #1 \right]}

\newcommand{\toas}{\xrightarrow{{\rm a.s.}}}

\newrefformat{eq}{(\ref{#1})}
\newrefformat{chap}{Chapter~\ref{#1}}
\newrefformat{sec}{Section~\ref{#1}}
\newrefformat{algo}{Algorithm~\ref{#1}}
\newrefformat{fig}{Fig.~\ref{#1}}
\newrefformat{tab}{Table~\ref{#1}}
\newrefformat{rmk}{Remark~\ref{#1}}
\newrefformat{clm}{Claim~\ref{#1}}
\newrefformat{def}{Definition~\ref{#1}}
\newrefformat{cor}{Corollary~\ref{#1}}
\newrefformat{lmm}{Lemma~\ref{#1}}
\newrefformat{prop}{Proposition~\ref{#1}}
\newrefformat{app}{Appendix~\ref{#1}}
\newrefformat{hyp}{Hypothesis~\ref{#1}}
\newrefformat{problem}{Problem~\ref{#1}}
\newrefformat{thm}{Theorem~\ref{#1}}

\newcommand{\pth}[1]{\left( #1 \right)}
\newcommand{\qth}[1]{\left[ #1 \right]}
\newcommand{\sth}[1]{\left\{ #1 \right\}}

\newcommand{\calC}{{\mathcal{C}}}
\newcommand{\calD}{{\mathcal{D}}}
\newcommand{\calE}{{\mathcal{E}}}
\newcommand{\calF}{{\mathcal{F}}}

\newcommand{\calH}{{\mathcal{H}}}
\newcommand{\calI}{{\mathcal{I}}}

\newcommand{\calL}{{\mathcal{L}}}

\newcommand{\calN}{{\mathcal{N}}}

\newcommand{\calS}{{\mathcal{S}}}

\newcommand{\calV}{{\mathcal{V}}}

\newcommand{\calX}{{\mathcal{X}}}

\newcommand{\conv}{{\sf Conv}}

\begin{document}

\title{Defending Non-Bayesian Learning against Adversarial Attacks
\thanks{This research is supported in part by National Science Foundation award NSF 1421918.
Any opinions, findings, and conclusions or recommendations expressed here are those of the authors
and do not necessarily reflect the views of the funding agencies or the U.S. government.}}

\author{Lili Su \hspace*{1in} Nitin H. Vaidya\\
Department of Electrical and Computer Engineering\\ University of Illinois at Urbana-Champaign\\
\{lilisu3, nhv\}@illinois.edu\\~\\
{\bf Contact author and e-mail}: Lili Su (lilisu3@illinois.edu)
}
\institute{}

\maketitle

~

\centerline{\bf Abstract}

\input{abstract_new}

~

\newpage

\input{intro_new}

\input{formulation_new}

%
%

\input{ByzantineConsensus_new}

\input{MainResults_new}
\input{failure-free}

\input{MBFL}

\section{Conclusion}
\label{sec:conclusion}
This paper addresses the problem of consensus-based non-Bayesian learning over multi-agent networks when an unknown subset of agents may be adversarial (Byzantine).  
We propose two learning rules, and characterize the tight network identifiability condition for any consensus-based learning rule of interest to exist.  In our first update rule, each agent updates its local beliefs as (up to normalization) the product of (1) the likelihood of the {\em cumulative} private signals and (2) the weighted geometric average of the beliefs of its incoming neighbors and itself. Under reasonable assumptions on the underlying network structure and the global identifiability of the network, we show that all the non-faulty agents asymptotically agree on the true state almost surely. For the case when every agent is failure-free, we show that (with high probability) each agent's beliefs on the wrong hypotheses decrease at rate $O(\exp (-Ct^2))$, where $t$ is the number of iterations, and $C$ is a constant. 
In general when agents may be adversarial, network identifiability condition specified for the above learning rule scales poorly in $m$. In addition, the computation complexity per agent per iteration of this learning rule is forbiddingly high. 
Thus, we propose a modification of our first learning rule, whose complexity per iteration per agent is $O(m^2 n \log n)$. 
We show that this modified learning rule works under a much weaker global identifiability condition that is independent of $m$.

We so far focussed on synchronous system and static network, our results may be generalizable to asynchronous as well as time varying network.

Throughout this paper, we assume that consensus among non-faulty agents needs to be achieved. Although this is necessary for the family of consensus-based algorithms (by definition), this is not the case for the non-faulty agents to collaboratively learn the true state in general. Indeed, there is a tradeoff between the capability of the network to reach consensus and the tight condition of the network detectability.  For instance, if the network is disconnected, then information cannot be propagated across the connected components. Thus, the non-faulty agents in each connected component have to be able to learn the true state. We leave investigating the above tradeoff as future work.

\bibliographystyle{abbrv}

\input{biblist1}
%



\end{document}

%% file: abstract_new.tex
This paper addresses the problem of non-Bayesian learning over multi-agent networks, where agents repeatedly collect partially informative observations about an {\em unknown} state of the world, and try to collaboratively learn the true state. 
We focus on the impact of the adversarial agents on the performance of consensus-based non-Bayesian learning, where non-faulty agents combine local learning updates with consensus primitives.  In particular, we consider the scenario where an unknown subset of agents suffer Byzantine faults -- agents suffering Byzantine faults behave arbitrarily. \\



We propose two learning rules. 
\begin{itemize}
\item In our first update rule, each agent updates its local beliefs as (up to normalization) the product of (1) the likelihood of the {\em cumulative} private signals and (2) the weighted geometric average of the beliefs of its incoming neighbors and itself. Under reasonable assumptions on the underlying network structure and the global identifiability of the network, we show that all the non-faulty agents asymptotically agree on the true state almost surely. For the case when every agent is failure-free, we show that (with high probability) each agent's beliefs on the wrong hypotheses decrease at rate $O(\exp (-Ct^2))$, where $t$ is the number of iterations, and $C$ is a constant. 
\item In general when agents may be adversarial, network identifiability condition specified for the above learning rule scales poorly in the number of state candidates $m$.  In addition, the computation complexity per agent per iteration of this learning rule is forbiddingly high. 
Thus, we propose a modification of our first learning rule, whose complexity per iteration per agent is $O(m^2 n \log n)$, where $n$ is the number of agents in the network.  We show that this modified learning rule works under a much weaker network identifiability condition. In addition, this new condition is independent of $m$. 
 \end{itemize}
~

~

%
%
%

%% file: intro_new.tex
%
\section{Introduction}
\label{intro}
Decentralized hypothesis testing (learning) has received significant amount of attention \cite{chamberland2003decentralized,gale2003bayesian,jadbabaie2012non,Tsitsiklis1988,tsitsiklis1993decentralized,varshney2012distributed,wong2012stochastic}.
The traditional decentralized detection framework consists of a collection of spatially distributed sensors and a fusion center \cite{Tsitsiklis1988,tsitsiklis1993decentralized,varshney2012distributed}. The sensors independently collect {\em noisy} observations of the environment state, and send only {\em summary} of the private observations to the fusion center, where a final decision is made. In the case when the sensors directly send all the private observations, the detection problem can be solved using a centralized scheme. The above framework does not scale well, since  each sensor needs to be connected to the fusion center and full reliability of the fusion center is required, which may not be practical as the system scales.

Distributed hypothesis testing in the {\em absence} of fusion center is considered in \cite{gale2003bayesian,cattivelli2011distributed,jakovetic2012distributed,bajovic2012large}. In particular, Gale and Kariv \cite{gale2003bayesian} studied the distributed hypothesis testing problem in the context of social learning, where fully Bayesian belief update rule is studied.
Bayesian update rule is impractical in many applications due to memory and computation constraints of each agent.

To avoid the complexity of Bayesian learning,
a non-Bayesian learning framework that combines local Bayesian learning with distributed consensus was proposed by Jadbabaie et al. \cite{jadbabaie2012non}, and has attracted much attention \cite{jadbabaie2013information,nedic2014nonasymptotic,rad2010distributed,shahrampour2013exponentially,Lalitha2014,shahrampour2015finite,shahrampour2014distributed,molavi2015foundations}.
Jadbabaie et al. \cite{jadbabaie2012non} considered the general setting where external signals are observed during each iteration of the algorithm execution. Specifically, the belief of each agent is repeatedly updated as the arithmetic mean of its local Bayesian update and the beliefs of its neighbors -- combining iterative consensus algorithm with local Bayesian update. It is shown \cite{jadbabaie2012non} that, under this learning rule,
each agent learns the true state almost surely. The publication of \cite{jadbabaie2012non} has inspired significant efforts in designing and analyzing non-Bayesian learning rules with a particular focus on refining the fusion strategies and analyzing the (asymptotic and/or finite time) convergence rates of the refined algorithms \cite{jadbabaie2013information,nedic2014nonasymptotic,rad2010distributed,shahrampour2013exponentially,Lalitha2014,shahrampour2015finite,shahrampour2014distributed,molavi2015foundations}.
In this paper we are particularly interested in the log-linear form of the update rule, in which, essentially, each agent updates its belief as the geometric average of the local Bayesian update and its neighbors' beliefs \cite{rad2010distributed,jadbabaie2013information,nedic2014nonasymptotic,shahrampour2013exponentially,Lalitha2014,shahrampour2015finite,shahrampour2014distributed,molavi2015foundations}.
The log-linear form (geometric averaging) update rule is shown to converge exponentially fast \cite{jadbabaie2013information,shahrampour2013exponentially}.
Taking an axiomatic approach, the geometric averaging fusion is proved to be  optimal \cite{molavi2015foundations}.
An optimization-based interpretation of this rule is presented in
\cite{shahrampour2013exponentially}, using dual averaging method with properly chosen proximal functions.
 Finite-time convergence rates are investigated independently in \cite{nedic2014nonasymptotic,Lalitha2014,shahrampour2014distributed}.
 Both \cite{nedic2014nonasymptotic} and \cite{shahrampour2015finite} consider time-varying networks, with slightly different network models. Specifically, \cite{nedic2014nonasymptotic} assumes that the union of every consecutive $B$ networks is strongly connected, while \cite{shahrampour2015finite} considers random networks.
In this paper, we consider static networks for ease of exposition, although we believe that
our results can be easily generalized to time-varying networks.

\vskip 0.5\baselineskip

The prior work implicitly assumes that the networked agents are reliable in the sense that they correctly follow the specified learning rules. However, in some practical multi-agent networks, this assumption may not hold. For example, in social networks, it is possible that some agents are adversarial, and try to prevent the true state from being learned by the good agents.
Thus, this paper focuses on the fault-tolerant version the non-Bayesian framework proposed in \cite{jadbabaie2012non}. In particular, we assume that an unknown subset of agents may suffer Byzantine faults.

An agent suffering Byzantine fault may not follow the pre-specified algorithms/protocols, and {\em misbehave arbitrarily}. For instance, a faulty agent may lie to other agents (possibly non-consistently) about its own estimates.
In addition, a faulty agent is assumed to have a complete knowledge of the system, including the network topology, the local functions of all the non-faulty agents, the algorithm specification of the non-faulty agents, the execution of the algorithm, the local estimates of all the non-faulty agents, and contents of messages the other agents send to each other. Also, the faulty agents can potentially collaborate with each other to prevent the non-faulty agents from achieving their goal.
%
An alternative fault model, where some agents may unexpectedly cease computing and communicate with each other asynchronously, is considered in our companion work \cite{su2016asynchronous}.
The Byzantine fault-tolerance problem was introduced by Pease et al. \cite{PeaseShostakLamport} and has attracted intensive attention from researchers \cite{Dolev:1986:RAA:5925.5931,fekete1990asymptotically,LeBlanc2012,vaidya2012iterative,vaidya2014iterative,Mendes:2013:MAA:2488608.2488657}. Our goal is to design algorithms that enable all the non-faulty agents to learn the underlying true state.

The existing non-Bayesian learning algorithms \cite{jadbabaie2013information,Lalitha2014,molavi2015foundations,nedic2014nonasymptotic,rad2010distributed,shahrampour2013exponentially,shahrampour2014distributed,shahrampour2015finite} are not robust to Byzantine agents, since the malicious messages sent by the Byzantine agents are indiscriminatingly utilized in the local belief updates.
On the other hand, the incorporation of Byzantine consensus is non-trivial, since (i) the {\em effective} communication networks are {\em dependent} on the all the random local observations, making it non-trivial to adapt analysis of previous algorithms to our setting; and (ii) the problem of identifying tight topological condition for reaching Byzantine multi-dimensional consensus iteratively is open, making it challenging to identify the minimal detectability condition on the networked agents to learn the true environmental state.

\paragraph{\bf Contributions:}
Our contributions are two-fold.

\begin{itemize}
\item We first propose an update rule wherein each agent iteratively updates its local beliefs as (up to normalization) the product of (1) the likelihood of the {\em cumulative} private signals and (2) the weighted geometric average of the beliefs of its incoming neighbors and itself (using iterative Byzantine multi-dimensional consensus). In contrast to the existing algorithms \cite{nedic2014nonasymptotic,Lalitha2014}, where only the {\em current} private signal is used in the update, our proposed algorithm relies on the {\em cumulative} private signals.  Under reasonable assumptions on the underlying network structure and the global identifiability of the network, we show that all the non-faulty agents asymptotically agree on the true state almost surely.
In addition, for the special case when every agent is guaranteed to be failure-free, we show that (with high probability) each agent's beliefs on the wrong hypotheses decrease at rate $O(\exp (-Ct^2))$, where $t$ is the number of iterations, and $C$ is a constant. Thus, our proposed rule may be of independent interest for the failure-free setting considered in \cite{jadbabaie2013information,Lalitha2014,molavi2015foundations,nedic2014nonasymptotic,rad2010distributed,shahrampour2013exponentially,shahrampour2014distributed,shahrampour2015finite}.
%
%
%
%
\item  The local computation complexity per agent of the first learning rule is high due to the adoption of multi-dimensional consensus primitives. More importantly, the network identifiability condition used for that learning rule scales poorly in the number of possible states $m$. 
    Thus, we propose a modification of our first learning rule, whose complexity per iteration per agent is $O(m^2 n \log n)$, where $n$ is the number of agents in the network.  We show that this modified learning rule works under a much weaker global identifiability condition, which is independent of $m$. We cast the general $m$--ary hypothesis testing problem into a collection of binary hypothesis testing sub-problems.
\end{itemize}

\paragraph{\bf Outline:}
The rest of the paper is organized as follows. Section \ref{prob formulation} presents the problem formulation.
Section \ref{Bconsensus} briefly reviews existing results on vector Byzantine consensus, and matrix representation of the state evolution. Our first algorithm and its correctness analysis are presented in Section \ref{main results}. Section \ref{failure-free} demonstrates the above learning rule in the special case when $f=0$, and presents a finite-time analysis.
The modified learning rule and its correctness analysis are summarized in Section \ref{modified}. 
Section \ref{sec:conclusion} concludes the paper, and discusses possible extensions.

%% file: formulation_new.tex
\section{Problem Formulation}
\label{prob formulation}
\paragraph{Network Model:}
Our network model is similar to the model used in \cite{DBLP:conf/sss/SuV15,vaidya2012iterative}.
We consider a synchronous system. A collection of $n$ agents (also referred as {\em nodes}) are connected by a {\em directed} network $G(\calV, \calE)$, where $\calV=\{1, \ldots, n\}$ and $\calE$ is the collection of {\em directed} edges. For each $i\in \calV$, let $\calI_i$ denote the set of incoming neighbors of agent $i$. In any execution, up to $f$ agents suffer Byzantine faults. For a given execution, let $\calF$ denote the set of Byzantine agents, and $\calN$ denote the set of non-faulty agents. Throughout this paper, we assume that $f$ satisfies the condition implicitly imposed by the given topology conditions mentioned later. We assume that each non-faulty agent knows $f$, but does not know the {\em actual} number of faulty agents $|\calF|$. \footnote{This is because the upper bound $f$ can be learned via long-time performance statistics, whereas, the actual size of $\calF$ varies across executions, and may be impossible to be predicted in some applications.}
Possible misbehavior of faulty agents includes sending incorrect and mismatching (or inconsistent) messages. The Byzantine agents are also assumed to have complete knowledge of system, including the network topology, underlying running algorithm, the states or even the entire history. The faulty agents may collaborate with each other adaptively \cite{Lynch:1996:DA:525656}.
Note that $|\calF|\le f$ and $|\calN|\ge n-f$ since at most $f$ agents may fail.
(As noted earlier, although we assume a static network topology, our results can be easily generalized to time-varying networks.)

Throughout this paper, we use the terms {\em agent} and {\em node} interchangeably.

\paragraph{Observation Model:} Our observation model is identical the model used in \cite{jadbabaie2012non,Lalitha2014,shahrampour2015finite}.
Let $\Theta=\{\theta_1, \theta_2, \ldots, \theta_m\}$ denote a set of $m$ environmental states, which we call {\em hypotheses}.
In the $t$-th iteration, each agent {\em independently} obtains a private signal about the environmental state $\theta^*$, which is initially unknown to every agent in the network.
Each agent $i$ knows the structure of its private signal, which is represented by a collection of parameterized marginal distributions $\calD^i=\{\ell_i(w_i | \theta)| \, \theta\in \Theta,\, w_i\in \calS_i\}$,
where $\ell_i(\cdot | \theta)$ is the distribution of private signal when $\theta$ is the true state, and $\calS_i$ is the finite private signal space. For each $\theta \in \Theta$, and each $i\in \calV$, the support of  $\ell_i(\cdot|\theta)$ is the whole signal space, i.e., $\ell_i(w_i|\theta)>0$, $\forall\, w_i\in \calS_i$ and $\forall\, \theta \in \Theta$.
Let $s_t^i$ be the private signal observed by agent $i$ in iteration $t$, and let ${\bf s}_t=\{s_t^1, s_t^2, \ldots, s_t^n\}$ be the signal profile at time $t$ (i.e., signals observed by the agents in iteration $t$). Given an environmental state $\theta$, the signal profile ${\bf s}_t$ is generated according to the joint distribution $\ell_1(s_t^1|\theta)\times \ell_2(s_t^2|\theta)\times \cdots \times \ell_n(s_t^n|\theta)$.
In addition, let $s^i_{1, t}$ be the signal history up to time $t$ for agent $i=1, \cdots, n$, and let ${\bf s}_{1,t}=\{s_{1,t}^1, s_{1,t}^2, \ldots, s_{1,t}^n\}$ be the signal profile history up to time $t$.

%
%

%% file: ByzantineConsensus_new.tex
\section{Byzantine Consensus}
\label{Bconsensus}
In this section, we briefly review relevant exsting results on Byzantine consensus.
Byzantine consensus has attracted significant amount of attention \cite{Dolev:1986:RAA:5925.5931,fekete1990asymptotically,vaidya2013byzantine,LeBlanc2012,vaidya2012iterative,vaidya2014iterative,Mendes:2013:MAA:2488608.2488657}.
While the past work mostly focus on scalar inputs, the more general vector (or multi-dimensional) inputs have been studied recently \cite{Mendes:2013:MAA:2488608.2488657,vaidya2013byzantine,vaidya2014iterative}.
Complete communication networks are considered in \cite{Mendes:2013:MAA:2488608.2488657,vaidya2013byzantine}, where tight conditions on the number of agents are identified.
Incomplete communication networks are studied in \cite{vaidya2014iterative}. Closer to the non-Bayesian learning problem is the class of {\em iterative approximate Byzantine consensus algorithms}, where each agent is only allowed to exchange information about its state with its neighbors.
In particular, our learning algorithms build upon  {\em Byz-Iter} algorithm proposed in \cite{vaidya2014iterative} and a simple algorithm proposed in \cite{vaidya2012iterative} for iterative Byzantine consensus with vector inputs and scalar inputs, respectively, in incomplete networks.
A matrix representation of the non-faulty agents' states evolution under  {\em Byz-Iter} algorithm is provided by \cite{vaidya2014iterative}, which also captures the dynamics of the simple algorithm with scalar inputs in \cite{vaidya2012iterative}.
%
To make this paper self-contained, in this section, we briefly review the algorithm {\em Byz-Iter} and its matrix representation.

\subsection{Algorithm {\em Byz-Iter} \cite{vaidya2014iterative}}
Algorithm {\em Byz-Iter} is based on Tverberg's Theorem \cite{Tverberg'sTheorem2007}.
\begin{theorem}\cite{Tverberg'sTheorem2007}
\label{TG}
Let $f$ be a nonnegative integer. Let $Y$ be a multiset containing vectors from $\reals^m$ such that $|Y|\ge (m+1)f+1$. There exists a partition $Y_1, Y_2, \cdots, Y_{f+1}$ of $Y$ such that $Y_i$ is nonempty for $1\le i\le f+1$, and the intersection of the convex hulls of $Y_i$'s are nonempty, i.e., $\cap_{i=1}^{f+1}\conv(Y_i)\not=\O$, where $\conv(Y_i)$ is the convex hull of $Y_i$ for $i=1, \cdots, f+1$.
\end{theorem}
The proper partition in Theorem \ref{TG}, and the points in $\cap_{i=1}^{f+1} \conv(Y_i)$, are referred as {\em Tverberg partition of $Y$} and {\em Tverberg points of $Y$}, respectively.


For convenience of presenting our algorithm in Section \ref{main results},
we present {\em Byz-Iter} (described in Algorithm \ref{BconsensusAA}) below using {\em One-Iter} (described in Algorithm \ref{BconsensusA}) as a primitive. 
The parameter ${\bf x}^i$ passed to {\em One-Iter} at agent $i$, and ${\bf y}^i$ returned
by {\em One-Iter} are both $m$-dimensional vectors.
Let ${\bf v}^i$ be the state of agent $i$ that will be iteratively updated, with ${\bf v}_t^i$ being the state at the end of iteration $t$ and ${\bf v}_0^i$ being the input of agent $i$. In each iteration $t\ge 1$, a non-faulty agent performs the steps in{\em One-Iter}. In particular, in the message receiving step, if a message is not received from some neighbor, that neighbor must be faulty, as the system is synchronous.  In this case, the missing message values are set to some default value. Faulty agents may deviate from the algorithm specification arbitrarily.
In {\em Byz-Iter}, the value returned by {\em One-Iter} at agent $i$ is assigned to ${\bf v}_t^i$.
\begin{algorithm}
\caption{Algorithm {\em One-Iter} ~~with input ${\bf x}^i$ at agent $i$}
\label{BconsensusA}
\vskip 0.2\baselineskip
{\normalsize
$Z^i\gets \O$\;

Transmit ${\bf x}^i$ on all outgoing links\;
\vskip 0.2\baselineskip

Receive messages on all incoming links. {\small \color{OliveGreen}\% These message values form a multiset $R^i$ of size $|\calI_i|$.\%}
\vskip 0.2\baselineskip

\For {every $C\subseteq R^i\cup \{{\bf x}^i\}$ such that $|C|=(m+1)f+1$}
{add to $Z^i$ a {\em Tverberg point} of multiset $C$}
\vskip 0.2\baselineskip

Compute ${{\bf y}^i}$ as follows:~~ ${\bf y}^i\gets \frac{1}{1+|Z^i|} \pth{{\bf x}^i+\sum_{{\bf z}\in Z^i}{\bf z}}$\;
\vskip 0.2\baselineskip
Return ${\bf y}^i$\;
}
\end{algorithm}

\begin{algorithm}

\caption{Algorithm {\em Byz-Iter}~ \cite{vaidya2014iterative}:  ~~ $t$-th iteration at agent $i$
}
\label{BconsensusAA}
{\normalsize
\vskip 0.2\baselineskip
${\bf v}_t^i \gets$  {\em One-Iter}(${\bf v}_{t-1}^i$)\;
\vskip 0.2\baselineskip
}
\end{algorithm}

\begin{remark}
Note that for each agent $i\in \calN$, the computation complexity per iteration is
\begin{align*}
\Omega \pth{\binom{|R^i\cup \{{\bf x}^i\}|}{(m+1)f+1}} = \Omega \pth{\binom{|\calI_i|+1}{(m+1)f+1}}. 
\end{align*}
In the worst case, $||\calI_i|+1|=n$, and 
 \begin{align*}
  \Omega \pth{\binom{|\calI_i|+1}{(m+1)f+1}}=\Omega \pth{\binom{n}{(m+1)f+1}}= \Omega \pth{\pth{\frac{n}{\eexp}}^{(m+1)f+1}}.
\end{align*}
Since our first learning rule is based on Algorithm {\em Byz-Iter}, the computation complexity of our first proposed algorithm is also high. Nevertheless, our first learning rule contains our main algorithmic ideas. More importantly, this learning rule can be modified such that the computation complexity per iteration per agent is $O(m^2 n\log n)$. Specifically, the modified learning rule adopts the scalar Byzantine consensus instead of the $m$--dimensional consensus. This modified learning rule is optimal in the sense that it works under minimal network identifiability requirements.
\end{remark}

\subsection{Correctness of Algorithm {\em Byz-Iter}}

We briefly summarize the aspects of correctness proof of Algorithm \ref{BconsensusAA} from \cite{vaidya2014iterative} that are necessary for our subsequent discussion.
By using the Tverberg points in the update of ${\bf v}_t^i$ above, effectively, the extreme message values (that may potentially be sent by faulty agents) are trimmed away. Informally speaking, trimming certain messages can be viewed as
ignoring (or removing) incoming links that carry the outliers. \cite{vaidya2014iterative}
shows that the effective communication network thus obtained can be characterized by a ``reduced graph'' of $G(\calV, \calE)$, defined below. It is important to note that the non-faulty agents {\bf do not} know the identity of the faulty agents. 
\begin{definition}[$m$--dimensional reduced graph]
\label{reduced graph}
An $m$--dimensional reduced graph $\calH(\calN, \calE_{\calF})$ of $G(\calV, \calE)$ is obtained by (i) removing all faulty nodes $\calF$, and all the links incident on the faulty nodes $\calF$; and (ii) for each non-faulty node (nodes in $\calN$), removing  up to $mf$ additional incoming links.
\end{definition}
\begin{definition}
\label{source}
A source component in any given $m$--dimensional reduced graph is a strongly connected component (of that reduced graph),
which does not have any incoming links from outside that component.
\end{definition}
It turns out that the effective communication network is potentially time-varying (partly) due to time-varying
behavior of faulty nodes.
Assumption \ref{sufficient} below states a condition that is sufficient for reaching approximate Byzantine vector consensus using Algorithm \ref{BconsensusA} \cite{vaidya2014iterative}.
\begin{assumption}
\label{sufficient}
Every $m$--dimensional reduced graph of $G(\calV, \calE)$ contains a unique source component.
\end{assumption}
Let $\calC_m$ be the set of all the $m$--dimensional reduced graph of $G(\calV, \calE)$.
Define $\chi_m \triangleq  |\calC_m| $.
Since $G(\calV, \calE)$ is finite,  we have $\chi_m<\infty$. Let $\calH_m \in \calC_m$ be an $m$--dimensional reduced graph of $G(\calV, \calE)$ with source component $\calS_{\calH_m}$. Define
\begin{align}
\label{minimal source}
\gamma_m\triangleq \min_{\calH_m\in \calC_m} |\calS_{\calH_m}|,
\end{align}
i.e., $\gamma_m$ is the minimum source component size among all the $m$--dimensional reduced graphs. Note that $\gamma_m\ge 1$ if Assumption \ref{sufficient} holds for a given $m$.
\begin{theorem}\cite{vaidya2014iterative}
\label{correctness of Byz-Iter}
Suppose Assumption \ref{sufficient} holds for a given $m\ge 1$. Under Algorithm {\em Byz-Iter}, all the non-faulty agents (agents in $\calN$) reach consensus asymptotically, i.e.,
$\lim_{t\diverge} |{\bf v}_t^i-{\bf v}_t^j| =0, \forall \, i,j\in \calN. $
\end{theorem}
The proof of Theorem \ref{correctness of Byz-Iter} relies crucially on a matrix representation of the state evolution. 

\subsection{Matrix Representation \cite{vaidya2014iterative}}
\label{mr}
Let $|\calF|=\phi$ (thus, $0\leq\phi\leq f$). Without loss of generality, assume that agents $1$ through $n-\phi$ are non-faulty, and agents $n-\phi+1$ to $n$ are Byzantine.

\begin{lemma}\cite{vaidya2014iterative}
\label{matrix lemma}
Suppose Assumption \ref{sufficient} holds for a given $m\ge 1$. The state updates performed by the non-faulty agents in the $t$--th iteration ($t\ge 1$) can be expressed as 
\begin{align}
\label{evo matrix}
{\bf v}_t^i=\sum_{j=1}^{n-\phi}{\bf A}_{ij}[t]{\bf v}^j_{t-1},
\end{align}
where ${\bf A}[t]\in \reals^{(n-\phi)\times (n-\phi)}$ is a {\em row stochastic} matrix for which there exists an $m$--dimensional reduced graph $\calH_m[t]$ with adjacency matrix ${\bf H}_m[t]$ such that $ {\bf A}[t]\ge \beta_m {\bf H}_m[t]$, where $0<\beta_m\le 1$ is a constant that depends only on $G(\calV, \calE)$.
\end{lemma}

Let ${\bf \Phi}(t,r) \triangleq  {\bf A}[t]\cdots {\bf A}[r]$ for $1\le r\le t+1$. By convention, ${\bf \Phi}(t,t)={\bf A}[t]$ and ${\bf \Phi} (t, t+1)={\bf I}$.
Note that ${\bf \Phi}(t,r)$ is a backward product. Using prior work
on coefficients of ergodicity \cite{Hajnal58}, under Assumption \ref{sufficient}, it has been shown \cite{vaidya2014iterative,wolfowitz1963products} that
\begin{align}
\label{mixing}
\lim_{t\ge r,~ t\diverge}{\bf \Phi}(t, r)=\ones {\bf \pi}(r),
\end{align}
where ${\bf \pi}(r)\in \reals^{n-\phi}$ is a row stochastic vector, and $\ones$ is the column vector with each entry being $1$.
Recall that $\chi_m$ is the total number of $m$--dimensional reduced graphs of $G(\calV, \calE)$, and $\beta_m$ is defined in Lemma \ref{matrix lemma},
and $\phi \triangleq  |\calF|$.
The convergence rate in \eqref{mixing} is exponential.
\begin{theorem}\cite{vaidya2014iterative}
\label{convergencerate}
For all $t\ge r\ge 1$, it holds that
$\left | {\bf \Phi}_{ij}(t, r)-\pi_j(r)\right |\le (1-\beta_m^\nu)^{\lceil\frac{t-r+1}{\nu}\rceil},$
where $\nu \triangleq  \chi_m(n-\phi)$.
\end{theorem}
Recall that $\gamma_m$ is defined in (\ref{minimal source}).
The next lemma is a consequence of the results in  \cite{vaidya2014iterative}.
\begin{lemma}\cite{vaidya2014iterative}
\label{lblimiting}
For any $r\ge 1$, there exists a reduced graph $\calH [r]$ with source component $\calS_r$ such that $\pi_i(r)\ge \beta_m^{\chi_m(n-\phi)}$ for each $i\in \calS_r$. In addition, $|\calS_r|\ge \gamma_m$.
\end{lemma}

\subsection{Tight Topological Condition for Scalar Iterative Byzantine Consensus}
The above analysis shows that Assumption \ref{sufficient} is sufficient for achieving Byzantine consensus iteratively. For the special case when $m=1$,(i.e., the inputs provided at individual non-faulty agents are scalars) it has been shown \cite{vaidya2012iterative} that Assumption \ref{sufficient} is also necessary.
\begin{theorem} \cite{vaidya2012iterative}
\label{iabc}
For scalar inputs, iterative approximate Byzantine consensus is achievable among non-faulty agents if and only if every $1$-dimensional reduced graph of $G(\calV, \calE)$ contains only one source component.
\end{theorem}
Moreover, the following simple algorithm (Algorithm \ref{BconsensusA scalar}) works under Assumption \ref{sufficient} when $m=1$.
\begin{algorithm}

\caption{Algorithm {Scalar Byzantine Consensus}: iteration $t\ge 1$  \cite{vaidya2012iterative}}
\label{BconsensusA scalar}
{\normalsize
Transmit $v^i[t-1]$ on all outgoing links\;
\vskip 0.2\baselineskip

Receive messages on all incoming links. {\small \color{OliveGreen}\% These message values $w_j[t]$ for each $j\in \calI_i$ form a multiset $R^i[t]$ of size $|\calI_i|$. \%}
\vskip 0.2\baselineskip

Sort the received values $w_j[t]$ for each $j\in \calI_i$ in a non-decreasing order\;
\vskip 0.2\baselineskip
Remove the largest $f$ values and the smallest $f$ values. {\small \color{OliveGreen}\%  Denote the set of indices of incoming neighbors whose values have not been removed at iteration $t$ by $\calI_i^*[t]$.\%}
\vskip 0.2\baselineskip

Update $v^i$ as follows: $v^i[t]\gets \frac{\sum_{j\in \calI_i^*[t]} w_j[t] + v^i[t-1]}{1+|\calI_i^*[t]|}$\;
}
\end{algorithm}

In addition, it has been show that the dynamic of the non-faulty agents states admits the same matrix representation as in Subsection \ref{mr} with the reduced graph being $1$--dimensional reduced graph defined in Definition \ref{reduced graph}.

With the above background on Byzantine vector consensus, we are now ready to present our first algorithm
and its analysis.

%% file: MainResults_new.tex
\section{Byzantine Fault-Tolerant Non-Bayesian Learning (BFL)}
\label{main results}

In this section, we present our first learning rule, named Byzantine Fault-Tolerant Non-Bayesian Learning (BFL). In BFL, each agent $i$ maintains a belief vector $\mu^i\in \reals^m$, which is a distribution over the set $\Theta$, with $\mu^i(\theta)$ being the probability with which the agent $i$ {\em believes} that $\theta$ is the true environmental state. Since no signals are observed before the execution of an algorithm, the belief $\mu^i$ is often initially set to be uniform over the set $\Theta$, i.e.,
$\pth{\mu_0^i(\theta_1), \mu_0^i(\theta_1), \ldots, \mu_0^i(\theta_m)}^T=\pth{\frac{1}{m}, \ldots, \frac{1}{m}}^T$. 
Recall that $\theta^*$ is the true environmental state.
We say the networked agents collaboratively learn $\theta^*$ if for every non-faulty agent $i\in \calN$,
\begin{align}
\label{goal}
\lim_{t\diverge}\mu_t^i(\theta^*) ~=~ 1 \,\,\, a.s. 
\end{align}
where $a.s.$ denotes {\em almost surely}.


BFL is a modified version of the geometric averaging update rule that has been investigated in previous work
\cite{nedic2014nonasymptotic,rad2010distributed,Lalitha2014,shahrampour2014distributed}.
In particular, we modify the averaging rule to take into account Byzantine faults.
More importantly, in each iteration, we use the likelihood of the {\em cumulative} local observations (instead of the likelihood of the {\em current} observation only) to update the local beliefs. 

\vskip 0.5\baselineskip

For $t\ge 1$, the steps to be performed by agent $i$ in the $t$--th iteration are listed below. 
Note that faulty agents can deviate from the algorithm specification.
The algorithm below uses {\em One-Iter} presented in the previous section as a primitive. Recall that $s_{1,t}^i$ is the cumulative local observations up to iteration $t$. Since the observations are $i.i.d.$, it holds that
$\ell_i(s_{1,t}^i|\theta)=\prod_{r=1}^t \ell_i(s_{r}^i |\theta)$. So $\ell_i(s_{1,t}^i|\theta)$ can be computed iteratively in Algorithm \ref{alg:new}.
%
%
\begin{algorithm}

\caption{BFL:~Iteration $t\geq 1$ at agent $i$}
\label{alg:new}
{\normalsize
$\eta_t^i\gets $ \,{\em One-Iter}$(\log \mu_{t-1}^i)$\;
Observe $s_t^i$\;

\For{$\theta\in\Theta$}
{$\ell_i(s_{1,t}^i|\theta)\gets \ell_i(s^i_t|\theta)\, \ell_i(s_{1,t-1}^i|\theta)$\;
$\mu_{t}^i(\theta)\gets \frac{\ell_i(s_{1, t}^i|\theta)\exp \pth{\eta_{t}^i(\theta)}}{\sum_{p=1}^m \ell_i(s_{1, t}^i|\theta_p)\exp \pth{\eta_{t}^i(\theta_p)}}$\;}
}
%
%
\end{algorithm}

The main difference of Algorithm \ref{alg:new} with respect to the algorithms in \cite{nedic2014nonasymptotic,rad2010distributed,Lalitha2014,shahrampour2014distributed}
is that (i) our algorithm uses a Byzantine consensus iteration as a primitive (in line 1), and (ii) $\ell_i(s_{1,t}^i|\theta)$ used in line 5 is the
likelihood for observations from iteration 1 to $t$ (the previous algorithms instead use $\ell_i(s_t^i |\theta)$ here). 
Observe that the consensus step is being performed on $\log$ of the beliefs, with the result being stored as $\eta_t^i$ (in line 1) and used in line 4 to compute the new beliefs.

Recalling the matrix representation of the {\em Byz-Iter} algorithm as per Lemma \ref{matrix lemma}, we can write the following equivalent representation of line 1 of Algorithm \ref{alg:new}.
\begin{eqnarray}
\eta_t^i(\theta) & = & \sum_{j=1}^{n-\phi} {\bf A}_{ij}[t]\log \mu_{t-1}^j(\theta)=\log \prod_{j=1}^{n-\phi}\mu_{t-1}^j(\theta)^{{\bf A}_{ij}[t]}, ~~~~\forall \theta\in \Theta.
\label{e:eta}
\end{eqnarray}
where ${\bf A}[t]$ is a row stochastic matrix whose properties are specified in Lemma \ref{matrix lemma}.
Note that $\mu_{t}^i(\theta)$ is {\bf random} for each $i\in \calN$ and $t\ge 1$, as it is updated according to local random observations. Since the consensus is performed over $\log \mu_{t}^i\in \reals^m$, the update matrix ${\bf A}[t]$ is also {\bf random}. In particular, for each $t\ge 1$, matrix ${\bf A}[t]$ is dependent on {\em all the cumulative observations over the network} up to iteration $t$.  This dependency makes it non-trivial to adapt analysis from previous algorithms to our setting. In addition, adopting the local cumulative observation likelihood makes the analysis with Byzantine faults easier.

\subsection{Identifiability}
In the absence of agent failures \cite{jadbabaie2012non}, for the networked agents to detect the true hypothesis $\theta^*$, it is sufficient to assume that $G(\calV, \calE)$ is strongly connected, and that $\theta^*$ is globally identifiable. That is, for any $\theta\not=\theta^*$, there exists a node $j\in \calV$ such that the Kullback-Leiber divergence between the true marginal $\ell_j(\cdot |\theta^*)$ and the marginal $\ell_j(\cdot |\theta)$, denoted by $D \pth{\ell_j(\cdot |\theta^*)||\ell_j(\cdot |\theta)}$, is nonzero;
equivalently,
\begin{align}
\label{failure-freeidentify}
\sum_{j\in \calV} D \pth{\ell_j(\cdot |\theta^*)||\ell_j(\cdot |\theta)}~\not=~0,
\end{align}
where $D \pth{\ell_j(\cdot |\theta^*)||\ell_j(\cdot |\theta)}$ is defined as
\begin{align}
\label{KL}
D \pth{\ell_j(\cdot |\theta^*)||\ell_j(\cdot |\theta)}\triangleq \sum_{w_j\in \calS_j}\ell_j(w_j|\theta^*)\log \frac{\ell_j(w_j|\theta^*)}{\ell_j(w_j|\theta)}.
\end{align}

Since $\theta^*$ may change from execution to execution, \eqref{failure-freeidentify} is required to hold for any choice of $\theta^*$.
Intuitively speaking, if any pair of states $\theta_1$ and $\theta_2$ can be distinguished by at least one agent in the network, then sufficient exchange of local beliefs over strongly connected network will enable every agent distinguish $\theta_1$ and $\theta_2$. However, in the presence of Byzantine agents, 
a stronger global identifiability condition is required. The following assumption builds upon Assumption \ref{sufficient}.
\begin{assumption}
\label{ass}
Suppose that Assumption \ref{sufficient} holds for $m=|\Theta|$. For any $\theta\not=\theta^*,$ and for any $m$--dimensional reduced graph $\calH$ of $G(\calV, \calE)$ with $\calS_{\calH}$ denoting the unique source component, the following holds
\begin{align}
\label{failure identify}
\sum_{j\in \calS_{\calH}} D\pth{\ell_j(\cdot |\theta^*)\parallel\ell_j(\cdot |\theta)}~\not=~0.
\end{align}
\end{assumption}
In contrast to \eqref{failure-freeidentify}, where the summation is taken over all the agents in the network, in \eqref{failure identify},  the summation is taken over agents in the source component only. Intuitively, the condition imposed by Assumption \ref{ass} is that all the agents in the source component can detect the true state $\theta^*$ collaboratively. If iterative consensus is achieved, the accurate belief can be propagated from the source component to every other non-faulty agent in the network.

\begin{remark}
We will show later that when Assumption \ref{ass}  holds, BFL algorithm enables all the non-faulty agents concentrate their beliefs on the true state $\theta^*$ almost surely. That is, Assumption \ref{ass} is a sufficient condition for a consensus-based non-Bayesian learning algorithm to exist. However, Assumption \ref{ass} is not necessary, observing that Assumption \ref{sufficient}  (upon which Assumption \ref{ass} builds) is not necessary for $m$-dimensional Byzantine consensus algorithms to exist. As illustrated by our second learning rule (described later), the adoption of $m$-dimensional Byzantine consensus primitives is {\em not necessary}. Nevertheless, BFL contains our main algorithmic and analytical ideas. In addition, BFL provides an alternative learning rule for the failure-free setting (where no fault-tolerant consensus primitives are needed).
\end{remark}

\subsection{Convergence Results}
Our proof parallels the structure of a proof in \cite{nedic2014nonasymptotic}, but with some key differences to take into account our update rule for the belief vector.\\

For any $\theta_1, \theta_2\in \Theta$, and any $i\in \calV$, define $\bm{\psi}_{t}^i(\theta_1, \theta_2)$ and $\calL_{t}(\theta_1, \theta_2)$ as follows
\begin{align}
\label{b1}
\bm{\psi}_{t}^i(\theta_1, \theta_2)\triangleq \log \frac{\mu_t^i(\theta_1)}{\mu_t^i(\theta_2)}, \quad \calL^i_{t}(\theta_1, \theta_2)~\triangleq~\log \frac{\ell_i(s_{t}^i|\theta_1)}{\ell_i(s_{t}^i|\theta_2)}.
\end{align}
To show Algorithm \ref{alg:new} solves \eqref{goal}, we will show that  $\bm{\psi}_{t}^i(\theta, \theta^*)\toas -\infty$ for $\theta\not=\theta^*$, which implies that $\mu_t^i(\theta)\toas 0$ for all $\theta\not=\theta^*$ and for all $i\in \calN$, i.e.,  all non-faulty agents asymptotically concentrate their beliefs on the true hypothesis $\theta^*$. We do this by investigating the dynamics of beliefs which is represented compactly in a matrix form.

For each $\theta\not=\theta^*$, and each $i\in \calN=\{1, 2, \cdots, n-\phi\}$, we have
\begin{align}
\label{rewrite1}
\nonumber
\bm{\psi}_{t}^i(\theta, \theta^*)&=\log \frac{\mu_{t}^i(\theta)}{\mu_{t}^i(\theta^*)}\overset{(a)}{=}\log \pth{\prod_{j=1}^{n-\phi} \pth{\frac{\mu_{t-1}^j(\theta)}{\mu_{t-1}^j(\theta^*)}}^{{\bf A}_{ij}[t]}\times \frac{\ell_i(s_{1,t}^i|\theta)}{\ell_i(s_{1,t}^i|\theta^*)}}\\
\nonumber
&=\sum_{j=1}^{n-\phi} {\bf A}_{ij}[t] \log \frac{\mu_{t-1}^j(\theta)}{\mu_{t-1}^j(\theta^*)} +\log \frac{\ell_i(s_{1,t}^i|\theta)}{\ell_i(s_{1,t}^i|\theta^*)}\\
&=\sum_{j=1}^{n-\phi} {\bf A}_{ij}[t] \bm{\psi}_{t-1}^j(\theta, \theta^*) +\sum_{r=1}^t\calL^i_{r}(\theta, \theta^*),
\end{align}
where equality (a) follows from \eqref{e:eta} and the update of $\mu^i$ in Algorithm \ref{alg:new}, and the last equality follows from \eqref{b1} and the fact that the local observations are $i.i.d.$ for each agent.

Let $\bm{\psi}_{t}(\theta, \theta^*)\in \reals^{n-\phi}$ be the vector that stacks $\bm{\psi}_t^i(\theta, \theta^*)$, with the $i$--th entry being $\bm{\psi}_t^i(\theta, \theta^*)$ for all $i\in \calN$.
The evolution of $\bm{\psi}(\theta, \theta^*)$ can be compactly written as
\begin{align}
\label{matrix form}
\bm{\psi}_{t}(\theta, \theta^*)&={\bf A}[t]\bm{\psi}_{t-1}(\theta, \theta^*)+\sum_{r=1}^t\calL_{r}(\theta, \theta^*).
\end{align}
Expanding \eqref{matrix form}, we get
\begin{align}
\label{int1}
\bm{\psi}_{t}(\theta, \theta^*)
={\bf \Phi}(t,1)\bm{\psi}_0(\theta, \theta^*)+\sum_{r=1}^{t} {\bf \Phi}(t,r+1)\sum_{k=1}^r\calL_k(\theta, \theta^*).
\end{align}
%
For each $\theta\in \Theta$ and $i\in \calV$, define $ H_i(\theta, \theta^*)\in \reals^{n-\phi}$  as
\begin{align}
\label{expected}
\nonumber
H_i(\theta, \theta^*)&\triangleq \sum_{w_i\in \calS_i} \ell_i(w_i| \theta^*) \log \frac{\ell_i(w_i\mid \theta)}{\ell_i(w_i\mid \theta^*)}\\
\nonumber
&= -D(\ell_i(\cdot |\theta^*)\parallel \ell_i(\cdot |\theta)) ~~~\text{by \eqref{KL}}\\
&\le 0.
\end{align}
Let $\calH\in \calC$ be an arbitrary reduced graph with source component $\calS_{\calH}$. Define $C_0$ and $C_1$ as
\begin{align}
-C_0&\triangleq  \min_{i\in \calV} \min_{\theta_1, \theta_2\in \Theta; \theta_1\not= \theta_2} \min_{w_i\in \calS_i} \pth{\log \frac{\ell_i(w_i|\theta_1)}{\ell_i(w_i|\theta_2)}}, \label{c0}\\
C_1&\triangleq \min_{\calH\in \calC} \, \min_{\theta, \theta^* \in \Theta; \theta\not= \theta^*} \sum_{i\in \calS_{\calH}}
D(\ell_i(\cdot | \theta^*) \parallel \ell_i(\cdot | \theta)).
\label{c1}
\end{align}
The constant $C_0$ serves as an universal upper bound on $|\log \frac{\ell_i(w_i|\theta_1)}{\ell_i(w_i|\theta_2)}|$ for all choices of $\theta_1$ and $\theta_2$, and for all signals. Intuitively, the constant $C_1$ is the minimal detection capability of the source component  under Assumption \ref{ass}.

Due to $|\Theta|=m< \infty$ and  $|\calS_i|< \infty$ for each $i\in \calN$, we know that $C_0<\infty$. Besides, it is easy to see that $-C_0\le 0$ (thus, $C_0\ge 0$). In addition, under Assumption \ref{ass}, we have $C_1>0$.
Now we present a key lemma for our main theorem.
\begin{lemma}
\label{second term goal}
Under  Assumption \ref{ass}, for any $\theta\not= \theta^*$, it holds that
\begin{align}
\frac{1}{t^2} \sum_{r=1}^{t}\pth{\sum_{j=1}^{n-\phi}{\bf \Phi}_{ij}(t,r+1)\sum_{k=1}^r\calL^j_k(\theta, \theta^*)- r\sum_{j=1}^{n-\phi}\pi_j(r+1)H_j(\theta, \theta^*)} \toas 0.
\end{align}
\end{lemma}
As it can be seen later, the proof of Lemma \ref{second term goal} is significantly different from the analogous lemma in \cite{nedic2014nonasymptotic}. 
\begin{theorem}
\label{almost sure}
When Assumption \ref{ass} holds, each non-faulty agent $i\in \calN$ will concentrate its belief on the true hypothesis $\theta^*$ almost surely, i.e., $\mu_t^i(\theta) \toas 0$ for all $\theta\not= \theta^*$.
\end{theorem}
\begin{proof}
Consider any $\theta\not=\theta^*$.
Recall from \eqref{int1} that
\begin{align*}
\bm{\psi}_{t}(\theta, \theta^*)&={\bf \Phi}(t,1)\bm{\psi}_0(\theta, \theta^*)+\sum_{r=1}^{t} {\bf \Phi}(t,r+1)\sum_{k=1}^r\calL_k(\theta, \theta^*)\\
&=\sum_{r=1}^{t} {\bf \Phi}(t,r+1)\sum_{k=1}^r\calL_k(\theta, \theta^*).
\end{align*}
The last equality holds as $\mu_0^i$ is uniform, and $\bm{\psi}^i_0(\theta, \theta^*)=0$ for each $i\in \calN$.
Since the supports of $\ell_i(\cdot|\theta)$ and $\ell_i(\cdot|\theta^*)$ are the whole signal space $\calS_i$ for each agent $i\in \calN$, it holds that $\left |\frac{\ell_i(w_i|\theta)}{\ell_i(w_i|\theta^*)}\right |<\infty$ for each $w_i\in \calS_i$, and
\begin{align}
\label{finite of H}
0\ge H_i(\theta, \theta^*)\ge \min_{w_i\in \calS_i} \pth{\log \frac{\ell_i(w_i|\theta)}{\ell_i(w_i|\theta^*)}}\ge ~-C_0 >-\infty.
\end{align}
By \eqref{finite of H}, we know that
$
|\sum_{j=1}^{n-\phi}\pi_j(r+1)  H_j(\theta, \theta^*)|\le C_0<\infty.
$
Due to the finiteness of $\sum_{j=1}^{n-\phi}\pi_j(r+1)  H_j(\theta, \theta^*)$, we are able to add and subtract $r \ones \sum_{j=1}^{n-\phi}\pi_j(r+1)  H_j(\theta, \theta^*)$ from \eqref{int1}. We get
\begin{align}
\label{int2}
\nonumber
\bm{\psi}_{t}(\theta, \theta^*)
&=\sum_{r=1}^{t}\pth{{\bf \Phi}(t,r+1)\sum_{k=1}^r\calL_k(\theta, \theta^*)-r \ones \sum_{j=1}^{n-\phi}\pi_j(r+1)  H_j(\theta, \theta^*) }\\
&\quad +\sum_{r=1}^t r \ones \sum_{j=1}^{n-\phi}\pi_j(r+1)  H_j(\theta, \theta^*).
\end{align}
For each $i\in \calN$, we have
\begin{align}
\label{evo}
\nonumber
\bm{\psi}^i_{t}(\theta, \theta^*)&=\sum_{r=1}^{t}\pth{\sum_{j=1}^{n-\phi}{\bf \Phi}_{ij}(t,r+1)\sum_{k=1}^r\calL^j_k(\theta, \theta^*)- r\sum_{j=1}^{n-\phi}\pi_j(r+1)H_j(\theta, \theta^*)}\\
&\quad+\sum_{r=1}^{t}  r\sum_{j=1}^{n-\phi}\pi_j(r+1)H_j(\theta, \theta^*).
\end{align}
To show $\lim_{t\diverge}\mu_t^i(\theta) \toas 0$ for $\theta\not= \theta^*$,  it is enough to show $\bm{\psi}^i_{t}(\theta, \theta^*)\toas -\infty$. Our convergence proof has similar structure as the analysis in \cite{nedic2014nonasymptotic}.
From Lemma \ref{second term goal}, we know that
\begin{align}
\label{lll2}
\frac{1}{t^2} \sum_{r=1}^{t}\pth{\sum_{j=1}^{n-\phi}{\bf \Phi}_{ij}(t,r+1)\sum_{k=1}^r\calL^j_k(\theta, \theta^*)- r\sum_{j=1}^{n-\phi}\pi_j(r+1)H_j(\theta, \theta^*)} \toas 0.
\end{align}
Next we show that the second term of the right hand side of \eqref{evo} decreases quadratically in $t$.
\begin{align}
\label{third term}
\nonumber
\sum_{r=1}^{t}  r\sum_{j=1}^{n-\phi} \pi_j(r+1)H_j(\theta, \theta^*)&\le \sum_{r=1}^{t}  r\sum_{j\in \calS_{r}} \pi_j(r+1)H_j(\theta, \theta^*)~~~~~~\text{by \eqref{expected}}\\
\nonumber
&\le   \sum_{r=1}^{t}  r  \beta^{\chi (n-\phi)} \sum_{j\in \calS_{r}}H_j(\theta, \theta^*)~~~~~~\text{by Lemma \ref{lblimiting}}\\
\nonumber
& \le - \sum_{r=1}^{t}  r  \beta^{\chi (n-\phi)} C_1~~~~~~~\text{by \eqref{c1} and \eqref{expected}}\\
&\le -\frac{t^2}{2}\beta^{\chi (n-\phi)} C_1.
\end{align}

Therefore, by \eqref{evo}, \eqref{lll2} and \eqref{third term}, almost surely, the following hold
\begin{align*}
\lim_{t\diverge}\frac{1}{t^2}\psi_t^i(\theta, \theta^*)
\le-\frac{1}{2}\beta^{\chi (n-\phi)} C_1.
\end{align*}
Therefore, we have $\psi_t^i(\theta, \theta^*) \toas -\infty$ and $\mu_t^i(\theta) \toas 0$ for $i\in \calN$ and $\theta\not=\theta^*$, proving Theorem \ref{almost sure}.

\eproof
\end{proof}

We now present the proof of our key lemma -- Lemma \ref{second term goal}.
\begin{proof}[Proof of Lemma \ref{second term goal}]
By \eqref{b1}, we have
\begin{align*}
\left |\calL^i_r(\theta, \theta^*)\right |=\left |\log \frac{\ell_i(s_{t}^i|\theta)}{\ell_i(s_{t}^i|\theta^*)}\right |\le \max_{i\in \calV} \max_{\theta_1, \theta_2 \in \Theta; \theta_1\not= \theta_2}
\max_{w_i\in \calS_i} \left |\log \frac{\ell_i(w_i|\theta_1)}{\ell_i(w_i|\theta_2)}\right |.
\end{align*}
Note that $\max_{i\in \calV} \max_{\theta_1, \theta_2 \in \Theta; \theta_1\not= \theta_2}
\max_{w_i\in \calS_i}  \left |\log \frac{\ell_i(w_i|\theta_1)}{\ell_i(w_i|\theta_2)}\right |$ is symmetric in $\theta_1$ and $\theta_2$. Thus,
\begin{align}
\label{d2}
\nonumber
\left |\calL^i_r(\theta, \theta^*)\right |&\le \max_{i\in \calV} \max_{\theta_1, \theta_2 \in \Theta; \theta_1\not= \theta_2}
\max_{w_i\in \calS_i}  \left |\log \frac{\ell_i(w_i|\theta_1)}{\ell_i(w_i|\theta_2)}\right |=\max_{i\in \calV} \max_{\theta_1, \theta_2 \in \Theta; \theta_1\not= \theta_2}
\max_{w_i\in \calS_i}  \log \frac{\ell_i(w_i|\theta_1)}{\ell_i(w_i|\theta_2)}\\
\nonumber
&=\max_{i\in \calV} \max_{\theta_1, \theta_2 \in \Theta; \theta_1\not= \theta_2}
\max_{w_i\in \calS_i}  -\log \frac{\ell_i(w_i|\theta_2)}{\ell_i(w_i|\theta_1)}\\
&=-\min_{i\in \calV} \min_{\theta_1, \theta_2 \in \Theta; \theta_1\not= \theta_2}
\min_{w_i\in \calS_i} \log \frac{\ell_i(w_i|\theta_2)}{\ell_i(w_i|\theta_1)}=-(-C_0)=C_0<\infty.
\end{align}
Thus, adding and subtracting $\frac{1}{t^2}\sum_{r=1}^t \sum_{j=1}^{n-\phi} \pi_j(r+1)\sum_{k=1}^r \calL_k^j(\theta, \theta^*)$ from the first term on the right hand side of \eqref{evo}, we can get
\begin{align}
\label{ltbb}
\nonumber
&\quad\frac{1}{t^2} \sum_{r=1}^{t}\pth{\sum_{j=1}^{n-\phi}{\bf \Phi}_{ij}(t,r+1)\sum_{k=1}^r\calL^j_k(\theta, \theta^*)-  \pi_j(r+1)r\sum_{j=1}^{n-\phi}H_j(\theta, \theta^*)}\\
\nonumber
&=\frac{1}{t^2} \sum_{r=1}^{t}\sum_{j=1}^{n-\phi}\pth{{\bf \Phi}_{ij}(t,r+1)- \pi_j(r+1)}\sum_{k=1}^r\calL^j_k(\theta, \theta^*)\\
&\quad + \frac{1}{t^2} \sum_{r=1}^{t}\sum_{j=1}^{n-\phi}  \pi_j(r+1)\pth{ \sum_{k=1}^r \calL_k^j(\theta, \theta^*)-rH_j(\theta, \theta^*)}.
\end{align}
For the first term of the right hand side of \eqref{ltbb}, we have
\begin{align}
\label{lt1}
\nonumber
\quad&\frac{1}{t^2} \left |\sum_{r=1}^{t}\sum_{j=1}^{n-\phi}\pth{{\bf \Phi}_{ij}(t,r+1)- \pi_j(r+1)}\sum_{k=1}^r\calL^j_k(\theta, \theta^*)\right |\\
&\le \frac{1}{t^2}\sum_{r=1}^{t}\sum_{j=1}^{n-\phi}\left |{\bf \Phi}_{ij}(t,r+1)- \pi_j(r+1)\right |\sum_{k=1}^r \left |\calL^j_k(\theta, \theta^*) \right |\\
\nonumber
&\le \frac{1}{t^2}\sum_{r=1}^{t}\sum_{j=1}^{n-\phi}\left |{\bf \Phi}_{ij}(t,r+1)- \pi_j(r+1)\right |r C_0~~~~\text{by \eqref{d2}}\\
\nonumber
&\le \frac{1}{t^2}\sum_{r=1}^{t}\sum_{j=1}^{n-\phi} (1-\beta^{\nu})^{\lceil\frac{t-r}{\nu}\rceil}r C_0~~~~\text{by Theorem \ref{convergencerate}}\\
\nonumber
&\le \frac{1}{t^2} \pth{t(n-\phi)C_0}\sum_{r=1}^{t}(1-\beta^{\nu})^{\lceil\frac{t-r}{\nu}\rceil} \\
&\le \frac{(n-\phi)C_0}{(1-\beta^{\nu})(1-(1-\beta^{\nu})^{\frac{1}{\nu}})t} .
\end{align}

Thus, for every sample path, we have
$$ \frac{1}{t^2} \sum_{r=1}^{t}\sum_{j=1}^{n-\phi}\pth{{\bf \Phi}_{ij}(t,r+1)- \pi_j(r+1)}\sum_{k=1}^r\calL^j_k(\theta, \theta^*) \to 0.$$

For the second term of the right hand side of \eqref{ltbb}, we will show that
$$\frac{1}{t^2} \sum_{r=1}^{t}\sum_{j=1}^{n-\phi}  \pi_j(r+1)\pth{ \sum_{k=1}^r \calL_k^j(\theta, \theta^*)-rH_j(\theta, \theta^*)} \toas 0,$$
i.e., almost surely for any $\epsilon>0$ there exists sufficiently large $t(\epsilon)$ such that $\forall \, t\ge t(\epsilon)$,
\begin{align}
\label{22}
\frac{1}{t^2} \left |\sum_{r=1}^{t}\sum_{j=1}^{n-\phi}  \pi_j(r+1)\pth{ \sum_{k=1}^r \calL_k^j(\theta, \theta^*)-rH_j(\theta, \theta^*)} \right |~\le ~\epsilon.
\end{align}
We prove this by dividing $r$ into two ranges $r\in \{1, \cdots, \sqrt{t}\}$ and $r\in \{\sqrt{t}+1, \cdots, t\}$, i.e.,
\begin{align}
\label{lll}
\nonumber
&\frac{1}{t^2}\sum_{r=1}^{t} \sum_{j=1}^{n-\phi} \pi_j(r+1)\pth{\sum_{k=1}^r \calL_k^j(\theta, \theta^*)-rH_j(\theta, \theta^*)}\\
\nonumber
&=\frac{1}{t^2}\sum_{r=1}^{\sqrt{t}} \sum_{j=1}^{n-\phi} \pi_j(r+1)\pth{\sum_{k=1}^r \calL_k^j(\theta, \theta^*)-rH_j(\theta, \theta^*)}\\
&\quad+\frac{1}{t^2}\sum_{r=\sqrt{t}+1}^t \sum_{j=1}^{n-\phi} \pi_j(r+1)\pth{\sum_{k=1}^r \calL_k^j(\theta, \theta^*)-rH_j(\theta, \theta^*)}.
\end{align}
For the first term of the right hand side of \eqref{lll}, we have
\begin{align*}
&\frac{1}{t^2}\left |\sum_{r=1}^{\sqrt{t}} \sum_{j=1}^{n-\phi} \pi_j(r+1)\pth{\sum_{k=1}^r \calL_k^j(\theta, \theta^*)-rH_j(\theta, \theta^*)}\right |\\
&\le \frac{1}{t^2} \sum_{r=1}^{\sqrt{t}} \sum_{j=1}^{n-\phi}\pi_j(r+1)\pth{2rC_0}~~~\text{by \eqref{expected} and \eqref{d2}}\\
&=\frac{1}{t^2} \pth{2C_0}  \sum_{r=1}^{\sqrt{t}}r\\
&\le C_0\pth{\frac{1}{t}+\frac{1}{t^{\frac{3}{2}}}}.
\end{align*}
Thus, there exists $t_1(\epsilon)$ such that for all $t\ge t_1(\epsilon)$, it holds that
\begin{align*}
\frac{1}{t^2}\left |\sum_{r=1}^{\sqrt{t}} \sum_{j=1}^{n-\phi}\pi_j(r+1)\pth{\sum_{k=1}^r \calL_k^j(\theta, \theta^*)-rH_j(\theta, \theta^*)}\right | \le \frac{\epsilon}{2}.
\end{align*}
For the second term of the right hand side of \eqref{lll}, we have
\begin{align*}
\quad&\frac{1}{t^2}\sum_{r=\sqrt{t}+1}^{t} \sum_{j=1}^{n-\phi}\pi_j(r+1)\pth{\sum_{k=1}^r \calL_k^j(\theta, \theta^*)-rH_j(\theta, \theta^*)}\\
&=\frac{1}{t}\sum_{r=\sqrt{t}+1}^{t} \sum_{j=1}^{n-\phi}\pi_j(r+1)\frac{r}{t}\pth{\frac{1}{r}\sum_{k=1}^r \calL_k^j(\theta, \theta^*)-H_j(\theta, \theta^*)}
\end{align*}

Since $\calL_k^j(\theta, \theta^*)$'s are i.i.d., from Strong LLN, we know that $\frac{1}{r}\sum_{k=1}^r \calL_k^j(\theta, \theta^*)-H_j(\theta, \theta^*)\toas 0$. That is, with probability 1, the sample path converges. Now, focus on each convergent sample path. For sufficiently large $r(\epsilon)$, it holds that for any $r\ge r(\epsilon)$,
$$\left| \frac{1}{r}\sum_{k=1}^r \calL_k^j(\theta, \theta^*)-H_j(\theta, \theta^*)\right | \le \frac{\epsilon}{2}.$$
Recall that $r\ge \sqrt{t}$. Thus, we know that there exists sufficiently large $t_2(\epsilon)$ such that $\forall \, t\ge t_2(\epsilon)$, $r\ge \sqrt{t}$ is large enough and
$$ \left |\frac{1}{r}\sum_{k=1}^r \calL_k^j(\theta, \theta^*)-H_j(\theta, \theta^*)\right |\le \frac{\epsilon}{2}.$$
Then, we have $\forall \, t\ge t_2(\epsilon)$,
\begin{align*}
&\frac{1}{t^2}\left |\sum_{r=\sqrt{t}+1}^{t} \sum_{j=1}^{n-\phi}\pi_j(r+1)\pth{\sum_{k=1}^r \calL_k^j(\theta, \theta^*)-rH_j(\theta, \theta^*)}\right |\\
&=\frac{1}{t}\sum_{r=\sqrt{t}+1}^{t} \sum_{j=1}^{n-\phi}\pi_j(r+1)\frac{r}{t}\left |\frac{1}{r}\sum_{k=1}^r \calL_k^j(\theta, \theta^*)-H_j(\theta, \theta^*)\right |\\
&\le \frac{1}{t}\sum_{r=\sqrt{t}+1}^{t} \sum_{j=1}^{n-\phi}\pi_j(r+1)\frac{r}{t} \frac{\epsilon}{2}\\
&=\frac{1}{t}\sum_{r=\sqrt{t}+1}^{t} \frac{r}{t} \frac{\epsilon}{2}=\frac{\epsilon}{2}\frac{1}{t^2} \sum_{r=\sqrt{t}+1}^{t}r \\
&=\frac{\epsilon}{4}\frac{1}{t^2} \pth{t^2-\sqrt{t}}\le ~\frac{\epsilon}{2}.
\end{align*}

Therefore, for any $\epsilon>0$, there exists $\max \{t_1(\epsilon), t_2(\epsilon)\}$, such that for any $t\ge \max \{t_1(\epsilon), t_2(\epsilon)\}$,
\begin{align*}
\frac{1}{t^2}\left |\sum_{r=1}^{t} \sum_{j=1}^{n-\phi}\pi_j(r+1)\pth{\sum_{k=1}^r \calL_k^j(\theta, \theta^*)-rH_j(\theta, \theta^*)}\right | ~\le ~\epsilon,
\end{align*}
for every convergent sample path.
In addition, we know a sample path is convergent with probability 1. Thus \eqref{22} holds almost surely.

Therefore, Lemma \ref{second term goal} is proved.

\eproof
\end{proof}

%% file: failure-free.tex
\section{BFL in the absence of Byzantine Agents, i.e., $f=0$}
\label{failure-free}

In this section, we present BFL for the special case in the absence of Byzantine agents, i.e., $f=0$, named Failure-free BFL. Since $f=0$, all the agents in the network are cooperative, and no trimming is needed. 
Indeed, the BFL for $f=0$ is a simple modification of the algorithm proposed in \cite{nedic2014nonasymptotic}.
\begin{algorithm}

\caption{Failure-free BFL}
\label{alg: failure-free}
{\normalsize
 \vskip 0.2\baselineskip
 Transmit current belief vector $\mu_{t-1}^i$ on all outgoing edges\;
 \vskip 0.2\baselineskip
 Wait until a private signal $s_t^i$ is observed and belief vectors are received from all incoming neighbors $\calI_i$\;
 %
  \vskip 0.2\baselineskip
\For{$\theta\in \Theta$}
{$\mu_{t}^i(\theta)\gets \frac{\ell_i(s_{1, t}^i|\theta)\prod_{j\in \calI_i\cup \{i\}} \mu_{t-1}^j(\theta)^{\frac{1}{|\calI_i|+1}}}{\sum_{p=1}^m \ell_i(s_{1, t}^i|\theta)\prod_{j\in \calI_i\cup \{i\}} \mu_{t-1}^j(\theta)^{\frac{1}{|\calI_i|+1}}}.$}
}
\end{algorithm}
%


For each time $t\ge 1$, we define a matrix that follows the structure of $G(\calV, \calE)$ as follows: 
\begin{align}
\label{matrix 1}
{\bf A}_{ij}\triangleq
\begin{cases}
\frac{1}{|\calI_i|+1}, & j\in \calI_i\cup \{i\}\\
0, & \text{otherwise. }
\end{cases}
\end{align}
Thus, the dynamic of $\psi_t^i (\theta, \theta^*) $ (defined in \eqref{b1}) under Algorithm \ref{alg: failure-free} can be written as
\begin{align*}
\psi_t^i (\theta, \theta^*) &= \log \frac{\mu_t^i(\theta)}{\mu_t^i(\theta^*)}\\
&=\log \frac{\ell_i(s_{1, t}^i|\theta)\prod_{j\in \calI_i\cup \{i\}} \mu_{t-1}^j(\theta)^{\frac{1}{|\calI_i|+1}}}{\ell_i(s_{1, t}^i|\theta^*)\prod_{j\in \calI_i\cup \{i\}} \mu_{t-1}^j(\theta^*)^{\frac{1}{|\calI_i|+1}}}\\
&=\log \prod_{j\in \calI_i \cup \{i\}}\qth{\frac{\mu_{t-1}^j(\theta)}{\mu_{t-1}^j(\theta^*)}}^{\frac{1}{|\calI_i|+1}}+\log \frac{\ell_i(s_{1, t}^i\mid \theta)}{\ell_i(s_{1, t}^i\mid \theta^*)}\\
&=\log \prod_{j\in \calI_i \cup \{i\}}\qth{\frac{\mu_{t-1}^j(\theta)}{\mu_{t-1}^j(\theta^*)}}^{\frac{1}{|\calI_i|+1}}+\sum_{r=1}^t\log \frac{\ell_i(s_{r}^i\mid \theta)}{\ell_i(s_{r}^i\mid \theta^*)}\\
&=\sum_{j=1}^n {\bf A}_{ij}\psi_{t-1}^i (\theta, \theta^*)+\sum_{r=1}^t \calL_r^i(\theta, \theta^*) ~~~\text{by \eqref{b1} and \eqref{matrix 1}}
\end{align*}

Recall that $\bm{\psi}_{t} (\theta, \theta^*)\in \reals^{n-\phi}$ is the vector that stacks $\psi_{t-1}^i (\theta, \theta^*)$ with the $i$--th entry being $\psi_{t-1}^i (\theta, \theta^*)$ for all $i\in \calN$. Since $f=0$, i.e., the network is free of failures, it holds that
$$ 0\le \phi=|\calF|\le f=0.$$
Thus, $\bm{\psi}_{t} (\theta, \theta^*)\in \reals^{n}$. Similar to \eqref{int1}, the evolution of $\bm{\psi}_{t} (\theta, \theta^*)$ can be compactly written as follows.
\begin{align}
\label{ff int1}
\nonumber
\bm{\psi}_{t}(\theta, \theta^*)&={\bf A}^t\bm{\psi}_0(\theta, \theta^*)+\sum_{r=1}^{t} {\bf A}^{t-r}\sum_{k=1}^r\calL_k(\theta, \theta^*)\\
&=\sum_{r=1}^{t} {\bf A}^{t-r}\sum_{k=1}^r\calL_k(\theta, \theta^*).
\end{align}
The last equality holds from the fact that $\bm{\psi}_0(\theta, \theta^*)=\zeros$.

As mentioned before, the non-Bayesian learning rules \cite{nedic2014nonasymptotic,rad2010distributed,Lalitha2014,shahrampour2014distributed} are consensus-based learning algorithms, wherein agents are required to reach a common decision asymptotically.
\begin{assumption}
\label{graph failure free}
The underlying communication network $G(\calV, \calE)$ is strongly connected.
\end{assumption}
It is easy to see that $G(\calV, \calE)$ itself is the only reduced graph of $G(\calV, \calE)$, and that Assumption \ref{graph failure free} is the special case of Assumption \ref{sufficient} when $f=0$. Thus,
$$ \chi_m=1, ~~~~~\text{and } ~~~~ \nu_m=\chi_m(n-\phi)=n.$$
Note that both $\chi_m$ and $\nu_m$ are independent of $m$ when $f=0$. Henceforth in this section, we drop the subscripts of $\chi_m$ and $\nu_m$ for ease of notation.

Similar to \eqref{mixing}, for any $r\ge 1$, we get
\begin{align*}
\lim_{t\ge r, ~ t\diverge} {\bf A}^{t-r} = \ones \bm{\pi}.
\end{align*}
Since ${\bf A}$ is time-invariant, the product limit $\lim_{t\ge r, ~ t\diverge} {\bf A}^{t-r} $ is also independent of $r$.

It is easy to see that
$$ {\bf A} \ge \frac{1}{n} {\bf H},$$
where ${\bf H}$ is the adjacency matrix of the communication graph $G(\calV, \calE)$, and that
\begin{align}
\label{ff ll}
 \pi_j \ge \frac{1}{n^n}, ~~~\forall\, j=1, \cdots, n.
 \end{align}

The following corollary is a direct consequence of Theorem \ref{convergencerate}, and its proof is omitted.
\begin{corollary}
\label{ff convergencerate}
For all $t\ge r\ge 1$, it holds that
$\left | [{\bf A}^{t-r}]_{ij}-\pi_j\right |\le (1-\frac{1}{n^n})^{\lceil\frac{t-r}{n}\rceil},$
where $[{\bf A}^{t-r}]_{ij}$ is the $i,j$--th entry of matrix ${\bf A}^{t-r}$.
\end{corollary}

In addition, when $f=0$, Assumption \ref{ass} becomes
\begin{assumption}
\label{ass failure-free}
Suppose that Assumption \ref{graph failure free} holds. For any $\theta\not=\theta^*,$ the following holds
\begin{align}
\label{failure-free identify}
\sum_{j=1}^m  D\pth{\ell_j(\cdot |\theta^*)\parallel\ell_j(\cdot |\theta)}~\not=~0.
\end{align}
\end{assumption}

As an immediate consequence of Theorem \ref{almost sure}, we have the following corollary.
\begin{corollary}
\label{almost sure ff}
When Assumption \ref{ass failure-free} holds, each agent $i$ will concentrate its belief on the true hypothesis $\theta^*$ almost surely, i.e., $\mu_t^i(\theta) \toas 0$ for all $\theta\not= \theta^*$.
\end{corollary}
Since Corollary \ref{almost sure ff} is the special case of Theorem \ref{almost sure} for $f=0$, the proof of Corollary \ref{almost sure ff} is omitted.

\subsection{Finite-Time Analysis of Failure-Free BFL}

In this subsection, we present the convergence rate that is achievable in finite time with high probability. Our proof is similar to the proof presented in \cite{nedic2014nonasymptotic,shahrampour2014distributed}.
\begin{lemma}
\label{expect}
Let $\lambda\triangleq \pth{1-(\frac{1}{n})^n}^{\frac{1}{n}}$, and let $\theta\not=\theta^*$, and consider $\psi_t^i(\theta, \theta^*)$ as defined in \eqref{b1}. Then, for each agent $i$ we have
$$\expect{\psi_t^i(\theta, \theta^*)} \le \frac{nC_0}{(1-\frac{1}{n^n})(1-\lambda)}t  -\frac{C_1}{2n^n}  t^2. $$
\end{lemma}
\begin{proof}

By \eqref{ff int1}, we have $\psi_{t}^i (\theta, \theta^*)=\sum_{r=1}^t \sum_{j=1}^n [{\bf A}^{t-r}]_{ij} \sum_{k=1}^r \calL_k^j(\theta, \theta^*). $
Taking expectation of $\psi_{t}^i (\theta, \theta^*)$ with respect to $\ell^i(\cdot \mid \theta^*)$, we get
\begin{align}
\label{ff exp}
\nonumber
\mathbb{E}^*\qth{\bm{\psi}^i_{t}(\theta, \theta^*)}&=\mathbb{E}^*\qth{
\sum_{r=1}^t \sum_{j=1}^n [{\bf A}^{t-r}]_{ij} \sum_{k=1}^r \calL_k^j(\theta, \theta^*)}\\
\nonumber
&=\sum_{r=1}^t \sum_{j=1}^n [{\bf A}^{t-r}]_{ij} \sum_{k=1}^r \mathbb{E}^*\qth{\calL_k^j(\theta, \theta^*)}\\
\nonumber
&=\sum_{r=1}^t \sum_{j=1}^n [{\bf A}^{t-r}]_{ij} r H_j(\theta, \theta^*)~~~\text{by ~~\eqref{expected}}\\
&=\sum_{r=1}^t \sum_{j=1}^n \pth{[{\bf A}^{t-r}]_{ij}-\pi_j} r H_j(\theta, \theta^*) +\sum_{r=1}^t \sum_{j=1}^n \pi_j r H_j(\theta, \theta^*).
\end{align}

For the first term in the right hand side of \eqref{ff exp}, we have
\begin{align}
\label{ff exp 1}
\nonumber
\sum_{r=1}^t \sum_{j=1}^n \pth{[{\bf A}^{t-r}]_{ij}-\pi_j} r H_j(\theta, \theta^*)&\le \sum_{r=1}^t \sum_{j=1}^n \left |[{\bf A}^{t-r}]_{ij}-\pi_j\right | r \left |H_j(\theta, \theta^*)\right |\\
\nonumber
&\le \sum_{r=1}^t \sum_{j=1}^n \qth{1-\frac{1}{n^n}}^{\lceil \frac{t-r}{n}\rceil} r C_0 ~~~\text{by Corollary \ref{ff convergencerate}, and  \eqref{c0}}\\
\nonumber
&=n C_0 \sum_{r=1}^t  \qth{1-\frac{1}{n^n}}^{\lceil \frac{t-r}{n}\rceil} r\\
&\le \frac{nC_0}{(1-\frac{1}{n^n})(1-\lambda)}t.
\end{align}

Since $G(\calV, \calE)$ is the only source component, $C_1$ (defined in \eqref{c1}) becomes
$$C_1 = \min_{\theta, \theta^* \in \Theta; \theta\not= \theta^*} \sum_{i=1}^n D(\ell_i(\cdot | \theta^*) \parallel \ell_i(\cdot | \theta)).$$
Thus, for the second term in the right hand side of \eqref{ff exp}, we get
\begin{align}
\label{ff exp 2}
\nonumber
\sum_{r=1}^t \sum_{j=1}^n \pi_j r H_j(\theta, \theta^*) &\le \sum_{r=1}^t \sum_{j=1}^n \frac{1}{n^n} r H_j(\theta, \theta^*)~~~\text{by \eqref{ff ll} and \eqref{expected}}\\
\nonumber
&=\frac{1}{n^n}\sum_{r=1}^t r\sum_{j=1}^n  H_j(\theta, \theta^*)\\
\nonumber
&\le -\frac{1}{n^n}\sum_{r=1}^t r C_1 \\
&\le -\frac{C_1}{2n^n}  t^2.
\end{align}

By \eqref{ff exp 1} and \eqref{ff exp 2}, \eqref{ff exp} becomes
\begin{align}
\nonumber
\mathbb{E}^*\qth{\bm{\psi}^i_{t}(\theta, \theta^*)}&=\sum_{r=1}^t \sum_{j=1}^n \pth{[{\bf A}^{t-r}]_{ij}-\pi_j} r H_j(\theta, \theta^*) +\sum_{r=1}^t \sum_{j=1}^n \pi_j r H_j(\theta, \theta^*)\\
&\le \frac{nC_0}{(1-\frac{1}{n^n})(1-\lambda)}t  - \frac{C_1}{2n^n}  t^2,
\end{align}
proving the lemma.

\eproof
\end{proof}

Similar to \cite{nedic2014nonasymptotic,shahrampour2014distributed}, we also use McDiarmid's Inequality.
\begin{theorem}[McDiarmid's Inequality]
\label{McInequality}
Let $X_1, \cdots, X_t$ be independent random variables and consider the mapping $H: \calX^t\to \reals$. If for $r=1, \cdots, t$, and every sample $x_1, \cdots, x_t$, $x_r^{\prime}\in \calX$, the function $H$ satisfies
$$ \left | H(x_1, \cdots, x_r, \cdots, x_t)-H(x_1, \cdots, x_r^{\prime}, \cdots, x_t)\right |\le c_r,$$
then for all $\epsilon>0$,
$$ \mathbb{P}\qth{|H(x_1, \cdots, x_t)-\mathbb{E}[H(x_1, \cdots, x_t)]|\ge \epsilon}\le \exp \sth{\frac{-2\epsilon^2}{\sum_{r=1}^t c_r^2}}.$$
\end{theorem}

\begin{theorem}
\label{finite time}
Under Assumption \ref{ass failure-free}, for any $\rho\in (0,1)$, there exists an integer $T(\rho)$ such that with probability $1-\rho$, for all $t\ge T(\rho)$ and for all $\theta\not= \theta^*$, we have
\begin{align*}
\mu_t^i(\theta)\le \exp \pth{ \frac{nC_0}{(1-\frac{1}{n^n})(1-\lambda)}t  -\frac{C_1}{4n^n} t^2 }
\end{align*}
where $C_0$ and $C_1$ are defined in \eqref{c0} and \eqref{c1} respectively, and
$ T(\rho)= \frac{64C_0^2n^{2n}}{3C_1^2}\log \frac{1}{\rho}.$
\end{theorem}

\begin{proof}
Since $\mu_{t}^i(\theta^*)\in (0,1]$, we have 
\begin{align*}
\mu_t^i(\theta)\le \frac{\mu_t^i(\theta)}{\mu_t^i(\theta^*)} =\exp \pth{\psi_t^i(\theta, \theta^*)}.
\end{align*}
Thus, we have
\begin{align*}
\mathbb{P}\pth{\mu_t^i(\theta)\ge \exp \pth{ \frac{nC_0}{(1-\frac{1}{n^n})(1-\lambda)}t  -\frac{C_1}{4n^n} t^2 }} &\le \mathbb{P}\pth{\psi_t^i(\theta, \theta^*)\ge  \frac{nC_0}{(1-\frac{1}{n^n})(1-\lambda)}t  -\frac{C_1}{4n^n} t^2}\\
&\le \mathbb{P}\pth{\psi_t^i(\theta, \theta^*)-\mathbb{E}^*\qth{\psi_t^i(\theta, \theta^*)}\ge \frac{C_1}{4n^n} t^2 }.
\end{align*}

Note that $\psi_t^i(\theta, \theta^*)$ is a function of the random vector $ {\bf s}_1, \cdots, {\bf s}_t$. For a given sample path ${\bf s}_1, \cdots, {\bf s}_t$, and for all $p\in \{1, \cdots, t\}$, we have
\begin{align*}
\quad&\max_{{\bf s}_p\in \calS_1\times \cdots \times \calS_t} \psi_t^i(\theta, \theta^*)-\min_{ {\bf s}_p\in \calS_1\times \cdots \times \calS_t} \psi_t^i(\theta, \theta^*)\\
&=\max_{{\bf s}_p\in \calS_1\times \cdots \times \calS_t} \sum_{r=1}^t \sum_{j=1}^n[{\bf A}^{t-r}]_{ij} \sum_{k=1}^r \calL_k(\theta, \theta^*)-\min_{{\bf s}_p\in \calS_1\times \cdots \times \calS_t} \sum_{r=1}^t \sum_{j=1}^n[{\bf A}^{t-r}]_{ij} \sum_{k=1}^r \calL_k(\theta, \theta^*)\\
&=\max_{{\bf s}_p\in \calS_1\times \cdots \times \calS_t} \sum_{r=p}^t \sum_{j=1}^n[{\bf A}^{t-r}]_{ij} \sum_{k=1}^r \calL_k(\theta, \theta^*)-\min_{{\bf s}_p\in \calS} \sum_{r=p}^t \sum_{j=1}^n[{\bf A}^{t-r}]_{ij} \sum_{k=1}^r \calL_k(\theta, \theta^*)\\
&=\max_{{\bf s}_p\in \calS_1\times \cdots \times \calS_t} \sum_{r=p}^t \sum_{j=1}^n[{\bf A}^{t-r}]_{ij} \calL_p(\theta, \theta^*)-\min_{{\bf s}_p\in \calS} \sum_{r=p}^t \sum_{j=1}^n[{\bf A}^{t-r}]_{ij}  \calL_p(\theta, \theta^*)\\
&\le \sum_{r=p}^t \sum_{j=1}^n[{\bf A}^{t-r}]_{ij} C_0 +\sum_{r=p}^t \sum_{j=1}^n[{\bf A}^{t-r}]_{ij}  C_0\\
&=2C_0(t-p+1)\triangleq c_{p}.
\end{align*}

By McDiarmid's inequality (Theorem \ref{McInequality}), we obtain that
\begin{align*}
\mathbb{P}\pth{\psi_t^*(\theta, \theta^*)-\mathbb{E}^*\qth{\psi_t^*(\theta, \theta^*)}\ge \frac{C_1}{4n^n} t^2 }&\le \exp \pth{-\frac{2\frac{C_1^2}{16n^{2n}}t^4}{\sum_{p=1}^t (2C_0(t-p+1))^2}}\\
&\le \exp \pth{-\frac{3C_1^2}{64C_0^2 n^{2n}}t},
\end{align*}

where the last inequality follows from the fact that
$$ t(t+1)(2t+1)\le 4t^3 ~~~\forall \, t\ge 2,$$
which can be shown by induction.

Therefore, for a given confidence level $\rho$, in order to have
$$\mathbb{P}\pth{\mu_t^i(\theta)\ge \exp \pth{\frac{nC_0}{(1-\frac{1}{n^n})(1-\lambda)}t  -\frac{C_1}{4n^n} t^2}}\le \rho, $$
 we require that
$$ t\ge T(\rho)= \frac{64C_0^2n^{2n}}{3C_1^2}\log \frac{1}{\rho}.$$

\eproof
\end{proof}

\begin{remark}
The above finite-time analysis is not directly applicable for the general case when $f>0$, due to the fact that the local beliefs are dependent on all the observations collected so far {\em as well as} all the future observations.
\end{remark}

\begin{remark}
Our analysis for the special when $f=0$ also works for time-varying networks \cite{nedic2014nonasymptotic}. 
In addition, with identical analysis, we are able to adapt the failure-free scheme to work in the more general setting where there is no underlying true state, and the goal is to have the agents collaboratively identify an optimal $\theta\in \Theta$ that best explains all the observations obtained over the whole network.

\end{remark}

%% file: MBFL.tex
\section{Modified BFL and Minimal Network Identifiability}
\label{modified}
%
To reduce the computation complexity per iteration in general, and to identify the minimal (tight) global identifiability of the network for any consensus-based learning rule of interest to learn the true state, we propose a modification of the above learning rule, which works under much weaker network topology and global identifiability condition.

We decompose the $m$-ary hypothesis testing problem into $m(m-1)$ (ordered) binary hypothesis testing problems. For each pair of hypotheses $\theta_1$ and $\theta_2$, each non-faulty agent updates the likelihood ratio of $\theta_1$ over $\theta_2$ as follows. Let $r^i_t(\theta_1, \theta_2)$ be the log likelihood ratio of $\theta_1$ over $\theta_2$ kept by agent $i$ at the end of iteration $t$. Our modified learning rule applies consensus procedures to log likelihood ratio, i.e., $r^i_t(\theta_1, \theta_2)$, which is a scalar. For Algorithm \ref{alg: pairwise}, we only require scalar iterative Byzantine (approximate) consensus among the non-faulty agents to be achievable. 

When scalar consensus is achievable, the following assumption on the identifiability of the network to detect $\theta^*$ is minimal, meaning that if this assumption is not satisfied, then no correct consensus-based non-Bayesian learning exists.
\begin{assumption}
\label{ass pairwise}
Suppose that every $1$-dimensional reduced graph of $G(\calV, \calE)$ contains only one source component.
For any $\theta\not=\theta^*,$ and for any $1$-dimensional reduced graph $\calH_1$ of $G(\calV, \calE)$ with $\calS_{\calH_1}$ denoting the unique source component, the following holds
\begin{align}
\label{failure identify pairwise}
\sum_{j\in \calS_{\calH_1}} D\pth{\ell_j(\cdot |\theta^*)\parallel\ell_j(\cdot |\theta)}~\not=~0.
\end{align}
\end{assumption}
Assumption \ref{ass pairwise} is minimal for the following reasons:  (1) For any consensus-based learning rule to work, the communication network $G(\calV, \calE)$ should support consensus with scalar inputs. That is, every $1$-dimensional reduced graph of $G(\calV, \calE)$ must contain only one source component. (2) Under some faulty behaviors of the Byzantine agents, one particular $1$--dimensional reduced graph may govern the entire dynamics of $r^i(\theta_1, \theta_2)$. If \eqref{failure identify pairwise} does not hold for that reduced graph, then the good agents may not able to distinguish $\theta_1$ from $\theta_2$.

\begin{algorithm}
\caption{Pairwise Learning}
\label{alg: pairwise}
 \vskip 0.2\baselineskip
 {\normalsize
 Initialization: \For{$\theta_1, \theta_2\in \Theta, \text{and}~ \theta_1\not=\theta_2$}
 {$r_0^i(\theta_1, \theta_2)\gets 0$\;}

\While{$t\ge 1$}{
\For{$\theta_1, \theta_2\in \Theta, \text{and}~ \theta_1\not=\theta_2$}
 {Transmit current belief vector $r_{t-1}^i(\theta_1, \theta_2)$ on all outgoing edges\;
\vskip 0.2\baselineskip
 Wait until a private signal $s_t^i$ is observed and log likelihood ratios $\tilde{r}_{t-1}^j(\theta_1, \theta_2)$ are received from all incoming neighbors $\calI_i$\;
\vskip 0.2\baselineskip
Sort the received log likelihood ratios $\tilde{r}_{t-1}^j(\theta_1, \theta_2)$ in a non-decreasing order, and remove the smallest $f$ values and the largest $f$ values. {\small \color{OliveGreen}\% Denote the set of indices of incoming neighbors whose ratios have not been removed at iteration $t$ by $\calI_i^*[t]$.\%}
\vskip 0.2\baselineskip
$r_{t}^i(\theta_1, \theta_2) \gets \frac{\sum_{j\in \calI_i^*[t]} \tilde{r}_{t-1}^j(\theta_1, \theta_2) + r_{t-1}^i(\theta_1, \theta_2)}{|\calI^*[t]|+1} + \log \frac{\ell_i (s^i_{1, t} \mid \theta_1) }{\ell_i (s^i_{1, t} \mid \theta_2)}.$
}
}
}
\end{algorithm}

For each iteration, the computation complexity per agent (non-faulty) can be calculated as follows. The cost-dominant procedure in each iteration is sorting the received log likelihood ratios, which takes $O(n\log n)$ operations.
In total, we have $m(m-1)$ order pairs of hypotheses. Thus, the total computation per agent per iteration is $O(m^2 n \log n)$.

\begin{theorem}
\label{pairwise learning}
Suppose Assumption \ref{ass pairwise} holds. Under Algorithm \ref{alg: pairwise}, for any $\theta \not=\theta^*$, the following holds:
\begin{align*}
r_{t}^i(\theta^*, \theta)\toas +\infty, ~\text{and }~~~ r_{t}^i(\theta, \theta^*)\toas -\infty.
\end{align*}
\end{theorem}

\begin{proof}
By  \cite{VaidyaMatrix2012}, we know that for each pair of hypotheses $\theta_1$ and $\theta_2$, there exists a row-stochastic matrix ${\bf M}^{1, 2}[t]\in \reals^{(n-\phi)\times (n-\phi)}$ such that
\begin{align}
\label{update pairwise}
r_{t}^i(\theta_1, \theta_2)= \sum_{j=1}^{n-\phi} {\bf M}^{1,2}_{ij}[t] r_{t-1}^j(\theta_1, \theta_2)+  \log \frac{\ell_i (s^i_{1, t} \mid \theta_1) }{\ell_i (s^i_{1, t} \mid \theta_2)}.
\end{align}
Note that matrix ${\bf M}^{1,2}$ depends on the choice of hypotheses $\theta_1$ and $\theta_2$.

For a given pair of hypotheses $\theta_1$ and $\theta_2$, let ${\bf r}_{t}(\theta_1, \theta_2)\in \reals^{n-\phi}$ be the vector that stacks $r_{t}^i(\theta_1, \theta_2)$. The evolution of ${\bf r}(\theta_1, \theta_2)$ can be compactly written as
\begin{align}
\label{update pairwise vector}
\nonumber
{\bf r}_{t}(\theta_1, \theta_2)&=  {\bf M}^{1,2}[t] {\bf r}_{t-1}(\theta_1, \theta_2)+  \sum_{r=1}^t \calL_r(\theta_1, \theta_2) \\
&=  \sum_{r=1}^t {\bf \Phi}^{1,2}(t, r+1) \sum_{k=1}^r \calL_k(\theta_1, \theta_2),
\end{align}
where ${\bf \Phi}^{1,2}(t, r+1)\triangleq {\bf M}^{1, 2}[t] {\bf M}^{1, 2}[t-1] \cdots {\bf M}^{1, 2}[r+1]$ for $r\le t$, ${\bf \Phi}^{1,2}(t, t)\triangleq {\bf M}^{1, 2}[t]$ and ${\bf \Phi}^{1,2}(t, t+1)\triangleq {\bf I}$.  We do the analysis for each pair of $\theta_1$ and $\theta_2$ separately.

The remaining proof is identical to the proof of Theorem \ref{almost sure}, and is omitted.

\eproof
\end{proof}

\begin{proposition}
\label{uniqueness}
Suppose there exists $\tilde{\theta}\in \Theta$ such that for any $\theta \not=\tilde{\theta}$, it holds that $r_{t}^i(\tilde{\theta}, \theta)\toas +\infty$, and $r_{t}^i(\theta, \tilde{\theta})\toas -\infty$. Then $\tilde{\theta}=\theta^*.$
\end{proposition}
\begin{proof}
We prove this proposition by contradiction. Suppose there exists $\tilde{\theta}\not=\theta^*\in \Theta$ such that for any $\theta \not=\tilde{\theta}$, it holds that $r_{t}^i(\tilde{\theta}, \theta)\toas +\infty$, and $r_{t}^i(\theta, \tilde{\theta})\toas -\infty$. Then we know that $r_{t}^i(\tilde{\theta}, \theta^*)\toas +\infty$ and $r_{t}^i(\theta^*, \tilde{\theta})\toas -\infty$, contradicting Theorem \ref{pairwise learning}. Thus, Proposition \ref{uniqueness} is true.

\eproof
\end{proof}